\documentclass[11pt]{article}
\usepackage{path,graphicx}
\usepackage{a4,amssymb, dsfont}
\usepackage{authblk}
\usepackage{amsmath}
\usepackage{lineno}
\usepackage{enumitem}
\usepackage{cite}
\usepackage[]{algorithm2e}
\usepackage[T1]{fontenc}
\usepackage[utf8]{inputenc}
\usepackage{authblk}
\usepackage[justification=centering]{caption}
\usepackage{etoolbox}
\patchcmd\thebibliography
 {\labelsep}
 {\labelsep\itemsep=0pt\relax}
 {}
 {\typeout{Couldn't patch the command}}
\newtheorem{definition}{Definition}[section]
\newtheorem{lemma}[definition]{Lemma}
\newtheorem{theorem}[definition]{Theorem}

\newtheorem{claim}[definition]{Claim}

\newcommand{\qed}{\hfill $\square$\smallskip}
\title{Target Set in Threshold Models}
\author{Ahad N. Zehmakan

abdolahad.noori@inf.ethz.ch}
\affil{Department of Computer Science, ETH Zurich}
\providecommand{\keywords}[1]{\textbf{\textit{Key Words:}} #1}
\date{} 
\begin{document}
\maketitle
\begin{abstract}
Consider a graph $G$ and an initial coloring, where each node is blue or red. In each round, all nodes simultaneously update their color based on a predefined rule. In a threshold model, a node becomes blue if a certain number or fraction of its neighbors are blue and red otherwise. What is the minimum number of nodes which must be blue initially so that the whole graph becomes blue eventually? We study this question for graphs which have expansion properties, parameterized by spectral gap, in particular the Erd\H{o}s-R\'{e}nyi random graph and random regular graphs.
\end{abstract}
\section{Introduction}
Assume that you are given a graph $G=(V,E)$ and an initial coloring, where each node is either blue or red. In discrete-time rounds, all nodes simultaneously update their color. In the \emph{$r$-threshold model} for some integer $r\ge 1$, a node becomes blue if it has at least $r$ blue neighbors and red otherwise. In the \emph{$\alpha$-threshold model} for some $0<\alpha<1$, a node becomes blue if at least $\alpha$ fraction of its neighbors are blue and red otherwise. (We assume that $r$ and $\alpha$ are fixed while we let $n$, the number of nodes in the underlying graph, tend to infinity.) 

In each of these two models, a set $T\subseteq V$ is called a \emph{target set} whenever the following holds: if $T$ is fully blue in some round, then the whole graph becomes blue eventually.\footnote{It is worth to stress that the concept of a target set, known also as \emph{dynamic monopoly}, can be defined more generally to capture similar models like \emph{bootstrap percolation} models, where a blue node remains blue forever, cf.~\cite{morrison2018extremal,hambardzumyan2017polynomial,zehmakan2018two}.} The minimum size of a target set has been studied extensively on different classes of graphs like lattice, hypercube, random graphs, planar graphs, regular graphs, and many more, cf.~\cite{balister2010random,gartner2018majority,peleg1997local,jeger2019dynamic,zehmakan2019tight,flocchini2004dynamic,gartner2017color}. We are interested in the graphs with good expansion properties.

There exist different parameters to measure the expansion of a graph. We consider an algebraic characterization of expansion. Assume that $A(G)$ is the adjacency matrix of graph $G=(V,E)$. Consider the normalized adjacency matrix $M=D^{-\frac{1}{2}}AD^{-\frac{1}{2}}$, where $D$ is the diagonal matrix such that the entries of the diagonal are the degrees of the nodes. Let $1=\lambda_1\ge \lambda_2\ge\cdots\ge\lambda_n\ge-1$ be the eigenvalues of $M$. We denote the second-largest absolute eigenvalue of the normalized matrix by $\sigma(G):=\max_{2\le i\le n}|\lambda_i|$. Graph $G$ has stronger expansion properties when $\sigma(G)$ is smaller. Here, we assume that $\sigma(G) <1$.

Consider the $\alpha$-threshold model on a $d$-regular $n$-node graph $G$. Let $B_t$ and $R_t$ denote the set of blue and red nodes in the $t$-th round and define $b_t:=|B_t|$, $r_t:=|R_t|$. We prove that if $b_0\ge\alpha n+\sqrt{2/(1-\alpha)}\,\sigma n$, the graph becomes fully blue and if $b_0\le \alpha n-\sqrt{2/\alpha}\,\sigma n$, it becomes fully red in $\mathcal{O}(\log_{1/\sigma^2}n)$ rounds.\footnote{This generalizes the results from~\cite{n2018opinion}.} Roughly speaking, for small values of $\sigma$ any set ``slightly'' larger (smaller) than $\alpha n$ is (not) a target set. Assume that initially each node keeps a piece of information, in our case one of the two colors, and we allow an adversary to corrupt almost half of the nodes. Our result asserts that in regular graphs with strong expansion properties if the nodes simply follow the $\alpha$-threshold model for $\alpha=1/2$ (that is, each node selects the majority color in its neighborhood, up to the tie-breaking rule), then they all will retrieve the original information in logarithmically many rounds. This is typically known as \emph{density classification property} and has application in distributed fault-local mending where redundant copies of
data are kept and the ``majority rule'' is applied to overcome the damage caused by failures, cf.~\cite{peleg1997local,peleg1998size}. We also study the minimum size of a target set in $r$-threshold model on $d$-regular graphs with good expansion properties. Moreover, we state that our results can be generalized to include irregular graphs.

From an algorithmic perspective, one might ask what is the minimum size of a target set for a given graph. It is known that this problem is NP-hard in both $r$-threshold and $\alpha$-threshold model for certain ranges of $\alpha$ and $r$, cf.~\cite{chen2009approximability,mishra2002hardness,mishra2006minimum}. 
We consider the same problem for the minimum size of a stable set. It is shown that this problem is NP-hard in $\alpha$-threshold and $r$-threshold model for any $0<\alpha<1$ and $r\ge 3$. In a graph $G$, we say a node set $S$ is a \emph{stable set} when the following holds: if $S$ is fully blue in some round, it remains blue forever, regardless of the color of all other nodes. A blue target set results in the full disappearance of red color while a blue stable set only guarantees the survival of blue color. 

We present our results regarding the $\alpha$-threshold and $r$-threshold model on regular expanders respectively in Sections~\ref{alpha-BP} and~\ref{r-BP}. The extension of these results to irregular graphs is given in Section~\ref{irregular}. Finally, we provide our hardness results concerning the minimum size of a stable set in Section~\ref{complexity}.
\section{Threshold Models on Expanders}
For a node $v$ and a node set $S$ in a graph $G=(V,E)$, define $d_S(v):=|\{u\in S: \{v,u\}\in E\}|$ and let $d(v)=d_{V}(v)$ be the \emph{degree} of $v$. We assume that $\delta(G)$ and $\Delta(G)$ denote respectively the minimum and maximum degree in $G$ and define $e(A,B):=|\{(v,u):v\in A, u\in B, \{v,u\}\in E\}|$. Now, let us provide Lemma~\ref{mixing}, which is our main tool in this section. Roughly speaking, it states that if $\sigma$ is small, then the number of edges among every two node sets is almost completely determined by their size.  
\begin{lemma} (Expander mixing lemma, cf.~\cite{hoory2006expander})
\label{mixing} For a $d$-regular graph $G=(V,E)$ and $A,B\subset V$
\[
|e(A,B)-\frac{|A|\,|B|d}{n}|\le\sigma d \,\sqrt{|A|\,|B|(1-\frac{|A|}{n})(1-\frac{|B|}{n})}\,.
\]
\end{lemma}
\subsection{$\mathbf{\alpha}$-Threshold Model}
\label{alpha-BP}

\begin{theorem}
\label{thm:alpha}
In the $\alpha$-threshold model on a $d$-regular graph $G=(V,E)$,

(i) if $b_0\le \alpha n-\sqrt{\frac{2}{\alpha}}\,\sigma n$, the graph becomes fully red in $\mathcal{O}(\log_{\frac{\alpha^2}{4\sigma^2}}n)$ rounds,

(ii) if $b_0\ge \alpha n+\sqrt{\frac{2}{1-\alpha}}\,\sigma n$, the graph becomes fully blue in $\mathcal{O}(\log_{\frac{(1-\alpha)^2}{4\sigma^2}}n)$ rounds.
\end{theorem}
Of course, the two aforementioned statements are sensible if respectively $\frac{\alpha^2}{4\sigma^2}>1$ and $\frac{(1-\alpha)^2}{4\sigma^2}>1$.

\begin{proof}
To prove part (i), let us first show that if $b_0\le \alpha n-\sqrt{2/\alpha}\,\sigma n$, then $b_1\le \frac{\alpha}{2}n$. Since each node $v\in B_1$ has at least $\alpha$ fraction of its neighbors in set $B_0$, we have that $e(R_0,B_1)\le \frac{1-\alpha}{\alpha}e(B_0,B_{1})$. Applying Lemma~\ref{mixing} yields
\[
\frac{r_0b_{1}d}{n}-\sigma d\sqrt{r_0b_1}\le \frac{1-\alpha}{\alpha}(\frac{b_{0}b_1d}{n}+\sigma d\sqrt{b_0b_1}).
\]
Multiplying by $\frac{n}{d\sqrt{b_1}}$ and rearranging the terms give us
\[
\sqrt{b_{1}}(r_0-\frac{1-\alpha}{\alpha}b_0)\le \sigma n(\sqrt{r_0}+\frac{1-\alpha}{\alpha}\sqrt{b_0}).
\]
Since $(\sqrt{r_0}+\frac{1-\alpha}{\alpha}\sqrt{b_0})^2\le ((1+\frac{1-\alpha}{\alpha})\sqrt{n})^2=\frac{n}{\alpha^2}$ and $r_0=n-b_0$, we get
\begin{equation*}
\label{eq 5}
b_1(n-(\frac{1-\alpha}{\alpha}+1)b_0)^2\le\frac{\sigma^2n^3}{\alpha^2} \Rightarrow b_1\le \frac{\sigma^2n^3}{\alpha^2(n-\frac{b_0}{\alpha})^2}\ .
\end{equation*}
The assumption of $b_0\le \alpha n-\sqrt{\frac{2}{\alpha}}\,\sigma n$ implies that
\[
(n-\frac{b_0}{\alpha})^2\ge (n-n+\sqrt{\frac{2}{\alpha^3}}\,\sigma n)^2=\frac{2}{\alpha^3}\,\sigma^2n^2.
\]
Therefore, we have
\[ 
b_1\le \frac{\sigma^2n^3}{\alpha^2(2/\alpha^3)\sigma^2n^2}=\frac{\alpha}{2}n.
\]
So far we proved that after one round there exist at most $\frac{\alpha}{2}n$ blue nodes. Now, we show that if $b_t\le \frac{\alpha}{2}n$ for $t\ge 0$, then $b_{t+1}\le\frac{4\sigma^2}{\alpha^2}b_t$. Thus, the graph becomes fully red in $\mathcal{O}(\log_{\frac{\alpha^2}{4\sigma^2}}n)$ rounds. Since each node in $B_{t+1}$ has at least $\alpha$ fraction of its neighbors in set $B_t$, we have $\alpha d\,b_{t+1}\le e(B_t,B_{t+1})$. Applying Lemma~\ref{mixing} to the right side yields
\[
\alpha d\, b_{t+1}\le\frac{b_tb_{t+1}d}{n}+\sigma d\,\sqrt{b_tb_{t+1}}\Rightarrow \sqrt{b_{t+1}}(\alpha-\frac{b_t}{n})\le\sigma\sqrt{b_t}.
\]
Utilizing $b_t\le \frac{\alpha}{2}n$ implies that
\[
b_{t+1}\le\frac{\sigma^2}{(\alpha-\frac{\alpha}{2})^2}b_t=\frac{4\sigma^2}{\alpha^2}b_t.
\]
Notice that $b_0\ge \alpha n+\sqrt{\frac{2}{1-\alpha}}\,\sigma n$ is equivalent to $r_0\le (1-\alpha)n-\sqrt{\frac{2}{1-\alpha}}\,\sigma n$. Thus, the proof of part (ii) follows by replacing blue with red and $\alpha$ with $1-\alpha$ in the above proof. \qed
\end{proof}
\paragraph{Tightness.} We show that the logarithmic upper bound in Theorem~\ref{thm:alpha} is asymptotically tight; that is, in this setting it might take $\Omega(\log n)$ rounds until the graph is fully red/blue. Assume that $\alpha=1/2$ and let $d$ be a large constant, say $d\ge 100$. Consider a $d$-regular graph $G$ with $\sigma\le 2/\sqrt{d}$, which is known to exist, and even can be constructed explicitly, cf.~\cite{lubotzky1988ramanujan}. Assume that initially all nodes are red except the nodes in distance at most $\ell=\frac{1}{2}\log_d n$ from an arbitrary node $v$. There are at most $d^{k}$ nodes in distance $k$ from $v$. Thus, the number of nodes which are blue initially is at most $\sum_{k=0}^{\ell}d^{k}\le d^{\ell+1}=\mathcal{O}(\sqrt{n})$. We have that 
\[
b_0=\mathcal{O}(\sqrt{n})\le \frac{n}{10}\stackrel{d\ge 100}{\le}\frac{n}{2}-2(\frac{2}{\sqrt{d}})\,n\stackrel{\alpha=1/2,\,\sigma\le 2/\sqrt{d}}{\le}\alpha n-\sqrt{\frac{2}{\alpha}}\,\sigma n.
\]
Thus, by Theorem~\ref{thm:alpha} the graph becomes fully red, but it takes at least $\ell=\Omega(\log n)$ rounds since node $v$ will remain red for the first $\ell$ rounds.

The random $d$-regular graph $\mathcal{G}_{n,d}$ is the random graph with a uniform distribution over all $d$-regular graphs on $n$ nodes, say $[n]=\{1,\cdots,n\}$. It is known~\cite{friedman2003proof} that $\sigma(\mathcal{G}_{n,d})\le2/\sqrt{d}$ for $d\ge 3$ asymptotically almost surely. (We say an event occurs asymptotically almost surely, a.a.s., if it occurs with probability $1-o(1)$ as $n$ tends to infinity.) Putting this statement in parallel with Theorem~\ref{thm:alpha} implies that in the $\alpha$-threshold model on $\mathcal{G}_{n,d}$ for $d\ge 3$, if $b_0\le \alpha n-\sqrt{8/\alpha d}\,n$ then the graph becomes fully red a.s.s. (See~\cite{gartner2018majority} for similar results on the special case of $\alpha=1/2$.)
\subsection{$\mathbf{r}$-Threshold Model}
\label{r-BP}
We first provide Lemma~\ref{lemma 1} and Lemma~\ref{lemma 2} about the structure of a regular graph. Building on these two lemmata, we prove in Theorem~\ref{thm:r} that in the $r$-threshold model on a $d$-regular graph $G$ there is a target set of size at most $2\beta n+1/\beta$ for $\beta=\frac{r}{(1-\sigma)d}$. (We could apply Theorem~\ref{thm:alpha} for $\alpha=r/d$ to find an upper bound of form $(r/d)n+\sqrt{2/(1-r/d)}\,\sigma n$. However as one can observe, this is a weaker bound for most choices of $r$, $d$ and $\sigma$.) 
\begin{lemma}
\label{lemma 1}
In the $r$-threshold model on a $d$-regular graph $G=(V,E)$, there is a stable set of size $s$ for $\beta n\le s\le 2\beta n+1/\beta$.
\end{lemma}
\paragraph{Proof sketch.} The main idea is to consider a partition of the nodes into sets of size roughly $\beta n$ such that the number of edges between the sets is minimized. Then, we show that a set of maximum size is stable otherwise we can move a node from that set to another one and reduce the number of edges in between, which results in a contradiction. The formal proof is given in the appendix, Section~\ref{stable}.
\begin{lemma}
\label{lemma 2}
In a $d$-regular graph $G=(V,E)$, for each set $S\subset V$ of size $s\ge\beta n$, there exists a node $v\in V\setminus S$ so that $d_S(v)\ge r$.
\end{lemma}
\begin{proof}
By Lemma~\ref{mixing}, we have
\begin{align*}
& \frac{s(n-s)d}{n}-\sigma d\,\sqrt{s(n-s)(\frac{n-s}{n})(\frac{s}{n})}\le e(S,V\setminus S)
\end{align*}
which yields
\[
\frac{s(n-s)(1-\sigma)d}{n}\le e(S,V\setminus S).
\]
Applying $s\ge \beta n>\frac{(r-1)n}{(1-\sigma)d}$ implies that $(r-1)(n-s)<e(S,V\setminus S)$. Thus, there exits a node $v\in V\setminus S$ such that $d_S(v)\ge r$ by the pigeonhole principle. \qed
\end{proof}
\begin{theorem}
\label{thm:r}
In the $r$-threshold model on a $d$-regular graph $G$, there exists a target set of size at most $2\beta n+1/\beta$.
\end{theorem}
\begin{proof}
Let $S$ be a stable set of size $s$ such that $\beta n\le s\le 2\beta n+1/\beta$, which must exist by Lemma~\ref{lemma 1}. We claim that if initially all nodes in set $S$ are blue, the whole graph becomes blue eventually. Since $S$ is a stable set, all nodes in $S$ stay blue forever. Furthermore by Lemma~\ref{lemma 2}, in each round at least one more node becomes blue until the whole graph is blue. (Note that by a simple inductive argument, in each round the nodes in set $S$ and the newly added nodes create a stable set and, thus, remain blue forever.) \qed 
\end{proof}
\subsection{Irregular Graphs}
\label{irregular}
To avoid unnecessary technicalities in the proofs, we limited ourselves to regular graphs so far. Now, we argue that our results from Sections~\ref{alpha-BP} and~\ref{r-BP} can be generalized to capture irregular graphs, by applying basically the same proof ideas. All we need to do is to apply a more general variant of Lemma~\ref{mixing}, cf.~\cite{hoory2006expander}, and replace $d$ by $\delta$ or $\Delta$, according to the case, in the proofs. Then, Theorem~\ref{thm:alpha} can be expressed more generally as following. In the $\alpha$-threshold model on a graph $G$ if $b_0\le\underline{b}$, then $G$ becomes fully red in $\mathcal{O}(\log_{\frac{\alpha^2\gamma^2}{4\sigma^2}}n)$ rounds and if $b_0\ge\overline{b}$, then it becomes fully blue in $\mathcal{O}(\log_{\frac{(1-\alpha)^2\gamma^2}{4\sigma^2}}n)$ rounds, where $\gamma=\frac{\delta}{\Delta}$,
\[
\underline{b}=\frac{\gamma^3}{\alpha\gamma^3+(1-\alpha)}\alpha n-\frac{\sqrt{2/\alpha}}{\alpha\gamma^3+(1-\alpha)}\,\sigma n
\]
and 
\[
\overline{b}
=\frac{1}{(1-\alpha)\gamma^3+\alpha}\alpha n+\frac{\sqrt{2/(1-\alpha)}}{(1-\alpha)\gamma^3+\alpha}\,\sigma n.
\]

Furthermore, in the $r$-threshold model on a graph $G$, there is a target set of size at most $2\beta'n+1/\beta'$, where $\beta'=\frac{r}{(1-\frac{\sigma}{\gamma})\delta}$. Note that in the special case of $\gamma=1$, this is equivalent to Theorem~\ref{thm:r} since $\beta'=\beta$.

In the Erd\H{o}s-R\'{e}nyi random graph $\mathcal{G}_{n,p}$ each edge is added independently with probability $p$ on a node set of size $n$, say $[n]$. 
By the results of  Le, Levina, and Vershynin~\cite{le2017concentration}, we know that $\sigma(\mathcal{G}_{n,p})=\mathcal{O}(1/\sqrt{np})$ a.a.s. for $p=\omega(\log n/n)$ (recall that $\log n/n$ is the connectivity threshold). Furthermore, by applying the Chernoff bound~\cite{feller1968introduction}, $(1-\epsilon)np\le d(v)\le (1+\epsilon)np$, for an arbitrarily small constant $\epsilon>0$, with probability $1-\exp(-\omega(\log n))$ for an arbitrary node $v$ in $\mathcal{G}_{n,p}$ if $p=\omega(\log n/n)$. The union bound implies that a.a.s. $1-\epsilon'\le \gamma$ for an arbitrarily small constant $\epsilon'$. By the results from above, in the $\alpha$-threshold model on $\mathcal{G}_{n,p}$ with $p=\omega(\log n/n)$, the minimum size of a target set is in $(1\pm\epsilon^{\prime\prime})\alpha n$ a.a.s. for an arbitrarily small constant $\epsilon^{\prime\prime}>0$.
\section{Complexity Results}
\label{complexity}
Let $MS_{\alpha}(G)$ and $MS_{r}(G)$ respectively denote the minimum size of a stable set in $\alpha$-threshold and $r$-threshold on a graph $G$.
\begin{theorem}
For any constant $0<\alpha<1$, the problem of determining $MS_{\alpha}(G)$ for a given graph $G$ is NP-hard. 
\end{theorem}
\begin{proof}
We provide a reduction from \textsc{$\alpha$-Clique}, which is the problem of deciding whether a given graph $G$ has a clique of size at least $\alpha n$ or not. \textsc{$\alpha$-Clique} is NP-hard for any constant $0<\alpha<1$ by a simple reduction from \textsc{Clique}, which is one of Karp's 21 NP-complete problems. For the sake of completeness, we present this reduction in the appendix, Section~\ref{clique}. 

Let $G=(V=\{v_1,\cdots,v_n\},E)$ be an instance of \textsc{$\alpha$-Clique}. We construct a graph $G'=(V',E')$, where $V':=\bigcup_{j=1}^{4}V^{(j)}\cup\{w_1,w_2,w_3,w_4\}$ for $V^{(j)}:=\{v^{(j)}_{i}:1\le i\le n\}$. For the edge set, assume that the induced subgraph $G'[V^{(j)}]$ is an empty graph for $j=2,3$ and the induced subgraph $G'[V^{(j)}]$ is a copy of $G$ for $j=1,4$. We also connect node $w_j$ to all nodes in $V^{(j)}$ for $1\le j\le 4$. Moreover, we connect nodes $v_{i}^{(1)}$, $v_{i'}^{(2)}$, and similarly nodes $v_{i}^{(3)}, v_{i'}^{(4)}$, for $1\le i\ne i'\le n$ if $\{v_i,v_{i'}\}\notin E$. Finally, we add an edge between $v_i^{(2)}, v_{i'}^{(3)}$ for $1\le i\ne i'\le n$ if $\{v_i,v_{i'}\}\in E$. 

We claim that $G$ has a clique of size at least $\alpha n$ if and only if $MS_{\alpha}(G')=\lceil \alpha n\rceil+1$. Before proving this claim, let us make the following simple, however very useful, observation. A node set $S$ is stable in the $\alpha$-threshold model if and only if each node $v$ in $S$ has at least $\lceil \alpha d(v)\rceil$ neighbors in $S$.

The first direction of the claim is quite straightforward. Assume that $G$ has a clique of size at least $\alpha n$. This implies that it includes a clique of size $\lceil \alpha n\rceil$.
There is a clique of the same size in the copy of $G$ on $V^{(1)}$. Let $S$ be the node set obtained by adding $w_1$ to the node set of this clique. Set $S$ is stable since $G'$ is $n$-regular by construction and each node $v\in S$ has at least $\lceil \alpha n\rceil$ neighbors in $S$. Note that by the above observation an $n$-regular graph cannot have a stable set of size $\lceil \alpha n\rceil$ or smaller. Thus, $MS_{\alpha}(G')=\lceil \alpha n\rceil+1$.

Assume that $MS_{\alpha}(G')=\lceil \alpha n\rceil+1$. This implies that $G'$ includes a stable set $S$ of size $\lceil \alpha n\rceil+1$. The induced subgraph by $S$ must be a clique since for $S$ to be stable each node in $S$ has to be connected to $\lceil \alpha n\rceil$ nodes in $S$. Now, we prove that $G$ includes a clique of size at least $\alpha n$. Set $S$ includes at most one node from $V^{(2)}$ (similarly $V^{(3)}$) because $G'[V^{(2)}]$ (respectively $G'[V^{(3)}]$) is an empty graph. Furthermore, it is not possible that $S$ has exactly one node $v\in V^{(2)}$ and one node $v'\in V^{(3)}$ because $S$ must include a node, say $v^{\prime\prime}$, from $V^{(1)}$ or $V^{(4)}$, which cannot be connected to both $v$ and $v^{\prime}$. It is also easy to see that $w_2$ and $w_3$ cannot be in $S$. Therefore by symmetry, we are left with the two following cases: (i) $S$ includes one node from $V^{(2)}$ and the rest of nodes from $V^{(1)}$ or (ii) all nodes from $V^{(1)}\cup\{w_1\}$. Both cases imply that a subset of nodes in $V^{(1)}$ induces a clique of size at least $\alpha n$; thus, $G$ has also a clique of the same size since $G'[V^{(1)}]$ is a copy of $G$. \qed
\end{proof}

We observe that in the $r$-threshold model on a graph $G=(V,E)$, a set $S$ is stable if and only the induced subgraph $G[S]$ has minimum degree at least $r$. Thus, $MS_1(G)=2$ if $E\ne \emptyset$. Furthermore, $MS_2(G)$ is equal to the length of the shortest cycle in $G$, i.e., the girth of $G$. Therefore, the problem of determining $MS_r(G)$ for a given graph $G$ is in P if $r=1,2$. However for $r\ge 3$, the problem is NP-hard.
\begin{theorem} (Amini, Peleg, Pérennes, Sau, and Saurabh~\cite{amini2012approximability})
The problem of determining the minimum size of a set whose induced subgraph has minimum degree at least $r$ in a given graph $G$ does not admit any constant-factor approximation for $r\ge 3$, unless P = NP.
\end{theorem} 
\bibliographystyle{acm}
\bibliography{refer} %
\appendix
\section{Proof of Lemma~\ref{lemma 1}}
\label{stable}
Let $G=(V,E)$ be a $d$-regular graph. We prove that in the $r$-threshold model on $G$, there is a stable set of size $\beta n\le s\le 2\beta n+1/\beta$, where $\beta=r/(1-\sigma)d
$. Assume that $\mathcal{P}$ is the set of all partitions of $V$ into $\lfloor 1/\beta\rfloor$ sets such that all sets are of size at least $\lfloor \beta n\rfloor$, except one set which can be of size $\lfloor \beta n\rfloor-1$. Let $P\in \mathcal{P}$ be a partition for which the number of edges between the sets is minimized. Let $V_{\max}$ be a set of maximum size in $P$. Clearly, $V_{\max}$ is at least of size $\beta n$ and at most of size
\[
n-(\lfloor \frac{1}{\beta}\rfloor-2)\lfloor \beta n\rfloor-(\lfloor\beta n\rfloor-1)=n-\lfloor \frac{1}{\beta}\rfloor\lfloor \beta n\rfloor+\lfloor \beta n\rfloor+1\le n-(\frac{1}{\beta}-1)(\beta n-1)+\beta n+1.
\]  
which is equal to $2\beta n+1/\beta$. Therefore, $\beta n\le |V_{\max}|\le 2\beta n+1/\beta$. Furthermore, we claim that for each node $v\in V_{\text{max}}$, $d_{V_{\text{max}}}(v)\ge r$, which implies that $V_{\max}$ is a stable set. Assume that there is a node $u$ which violates this property, i.e., $d_{V_{\max}}(u)\le r-1$. Then, the average number of edges between $u$ and the $\lfloor 1/\beta\rfloor-1$ other sets is at least
\[
\frac{d-(r-1)}{\lfloor 1/\beta\rfloor-1}\ge \frac{d-(r-1)}{\frac{(1-\sigma)d}{r}-1}>\frac{d-(r-1)}{\frac{d}{r-1}-1}=r-1.
\] 
Thus, there must exist a set $V'$ among the other $\lfloor 1/\beta\rfloor-1$ sets such that $d_{V'}(u)\ge r$. This is a contradiction because by removing $u$ from $V_{\max}$ and adding it into $V'$, the number of edges between the sets decreases at least by one and it is easy to see that the new partition is also in $\mathcal{P}$.
\section{\textsc{$\alpha$-Clique} Is NP-hard}
\label{clique}
We prove that {$\alpha$-Clique} is NP-hard by a reduction from \textsc{Clique}. This statement is known by prior work for the special case of $\alpha=1/2$, cf.~\cite{lu2013maximum}. Let us first define both problems formally.
\\
\\
\textsc{$\alpha$-Clique}\\
\textit{Instance}: Graph $G$.\\
\textit{Output}:~Does $G$ have a clique of size at least~$\alpha n$?
\\
\\
\textsc{Clique}\\
\textit{Instance}: Graph $G$ and integer $k$.\\
\textit{Output}:~Does~$G$~have~a~clique~of~size~at~least~$k$?

\begin{claim}
\textsc{$\alpha$-Clique} is NP-hard for any constant $0<\alpha<1.$
\end{claim}
\begin{proof}
Let $n$-node graph $G=(V,E)$ and integer $k$ be an instance of \textsc{Clique}. We construct a graph $G'=(V',E')$ with $n'$ nodes such that $G$ has a clique of size at least $k$ if and only if $G'$ has a clique of size at least $\alpha n'$. We distinguish two cases.
\begin{itemize}
\item If $k\ge \alpha n$, then we add $\frac{k}{\alpha}-n$ isolated nodes to $G$ to obtain $G'$. (For simplicity in calculations we assume that $\frac{k}{\alpha}-n$ is an integer, otherwise basically the same argument, with a bit of lengthier calculations, works for $\lfloor \frac{k}{\alpha}-n\rfloor$.) 

Assume that $G$ has a clique of size at least $k$, then there is a copy of this clique in $G'$ as well. We have that $\frac{k}{n'}=\alpha$ since $n'=n+\frac{k}{\alpha}-n=\frac{k}{\alpha}$. Therefore, $G'$ has a clique of size at least $\alpha n'$.

Assume that $G'$ has a clique of size at least $\alpha n'=\alpha\,\frac{k}{\alpha}=k$, then $G$ has a clique of the same size.

\item If $k<\alpha n$, to obtain $G'$, we add a node set $U$ of size $\frac{\alpha n-k}{1-\alpha}$ to $G$ and connect each node in $U$ to all other nodes in $U$ and all nodes in $G$. (Again for the sake of simplicity, we assume that $\frac{\alpha n-k}{1-\alpha}$ is an integer, but the same argument works for $\lceil\frac{\alpha n-k}{1-\alpha}\rceil$ in general case.) Thus, the number of nodes in $G'$ is equal to 
\[
n'=n+\frac{\alpha n-k}{1-\alpha}=\frac{(1-\alpha)n+\alpha n-k}{1-\alpha}=\frac{n-k}{1-\alpha}.
\]

Assume that $G$ has a clique of size $k$. There is a clique of the same size in the copy of $G$ in $G'$. By adding all nodes in $U$ to this clique, we get a clique of size $k+\frac{\alpha n-k}{1-\alpha}=\frac{\alpha(n-k)}{1-\alpha}=\alpha n'$.

If $G'$ has a clique of size at least $\alpha n'$, then at most $\frac{\alpha n-k}{1-\alpha}$ of the nodes can be from $U$ and the rest must be from the copy of $G$. Thus, $G$ contains a clique of size at least
\[
\alpha n'-\frac{\alpha n-k}{1-\alpha}=\frac{\alpha(n-k)}{1-\alpha}+\frac{k-\alpha n}{1-\alpha}=k.
\]  \qed
\end{itemize}
\end{proof}
\end{document}